\documentclass{stacs_proc}

\theoremstyle{plain}

\usepackage{epsfig}
\usepackage{amssymb}

\begin{document}

\title[A Theory for Valiant's Matchcircuits]{A Theory for Valiant's Matchcircuits\\
  (Extended Abstract)}

\author[lab1]{Angsheng Li}{Angsheng Li}
\address[lab1]{State Key Laboratory of Computer Science,
\newline  Institute of Software, Chinese Academy of Sciences,
\newline P.O.\ Box 8717, Beijing 100080, China}  
\email{angsheng@ios.ac.cn}  
\email{xmjljx@gmail.com} 

\author[lab1]{Mingji Xia}{Mingji Xia}

\thanks{Both authors are supported by the NSFC Grant
 no.\ 60325206 and no.\ 60310213. The second author is also
supported by MATHLOGAPS (MEST-CT-2004-504029).} 



\begin{abstract}
The computational function of a matchgate is represented by its
character matrix. In this article, we show that all nonsingular
character matrices are closed under matrix inverse operation, so
that for every $k$, the nonsingular character matrices of $k$-bit
matchgates form a group, extending the recent work of Cai and
Choudhary \cite{cc06c} of the same result for the case of $k=2$, and
that the single and the two-bit matchgates are universal for
matchcircuits, answering a question of Valiant \cite{val02}.
\end{abstract}

\maketitle

\stacsheading{2008}{491-502}{Bordeaux}
\firstpageno{491}

\section{Introduction}
Valiant \cite{val02} introduced the notion of  matchgate and
matchcircuit as a new model of computation to simulate quantum
circuits, and successfully realized a significant part of quantum
circuits by using this new model.  Valiant's new method organizes
certain computations based on the graph theoretic notion of perfect
matching and the corresponding algebraic object of the Pfaffian.
This leaves an interesting open question of characterizing the exact
power of the matchcircuits. To solve these problems, a significant
first step would be a better understanding the structures of the
matchgates and the matchcircuits, to which the present paper is
devoted.

In \cite{val04}, Valiant introduced the notion of {\it holographic
algorithm}, based on matchgates and their properties, but with some
additional ingredients of the choice of a set of linear basis
vectors, through which the computation is expressed and interpreted.

Matchgates and their character matrices have some nice properties,
which have already been  extensively studied.  In \cite{cc06c}, Cai
and Choudhary showed that  a matrix is the character matrix of a
matchgate if and only if it satisfies all the useful
Grassmann-Pl\"ucker identities, and all nonsingular character
matrices of two bits matchgates form a group.

 In the present paper, we show that for every $k$, all the
 nonsingular character matrices of
 $k$-bit matchgates
 form a group, extending the result of Cai and
Choudhary of the same
 result for the
case of $k=2$.

Furthermore, we show that every matchcircuit based on $k$-bit
matchgates for $k>2$ can be realized by a series of compositions of
either single bit or two bits matchgates. This result answers a
question raised by Valiant in \cite{val02}. The result is an analogy
of the quantum circuits in the matchcircuits \cite{qut}.

We organize the paper as follows. In section
\ref{sec:2,definitions}, we outline necessary definitions and
background of the topic. In section \ref{sec:tranformation T1-T4},
we state our results, and give some overview of the proofs. In
section \ref{sec4}, we establish our first result that for every
$k$, all nonsingular $k$-bit character matrices form a group. In
section \ref{sec5}, we prove the second result that level $2$
matchgates are universal for matchcircuits.

\section{Definitions} \label{sec:2,definitions}

\subsection{Graphs and Pfaffian}

Let $G=(V,E,W)$ be a weighted undirected graph, where
$V=\{1,2,\ldots,n\}$ is the set of vertices each represented by a
distinct positive integer, $E$ is the set of edges and $W$ is the
set of weights of the edges. We represent the graph by a
skew-symmetric matrix $M$, called the {\it skew-symmetric adjacency
matrix} of $G$, where $M(i,j)=w(i,j)$ if $i<j$, $M(i,j)=-w(i,j)$ if
$i>j$, and $M(i,i)=0$.

The {\it Pfaffian} of an $n\times n$ skew-symmetric matrix $M$ is
defined to be $0$ if $n$ is odd, $1$ if $n$ is $0$, and if $n=2k$
where $k>0$ then it is defined by
\begin{displaymath}
{\textrm Pf}(M)=\sum _\pi \epsilon_\pi w(i_1,i_2)w(i_3,i_4) \ldots
w(i_{2k-1},i_{2k}),
\end{displaymath}
where
\begin{itemize}
\item $\pi =[i_1,i_2,\ldots, i_{2k}]$  is a permutation on
$[1,2,\ldots, n]$,

\item the summation is over all permutations $\pi$, where
$i_1<i_2,i_3<i_4,\ldots,i_{2k-1}<i_{2k}$ and $i_1<i_3<\ldots
<i_{2k-1}$,

\item $\epsilon_\pi$ is the sign of the permutation $\pi$, or
equivalently,  $\epsilon_\pi$ is the sign or parity of the number of
overlapping pairs, where a pair of edges
$(i_{2r-1},i_{2r}),(i_{2s-1},i_{2s})$ is {\it overlapping} iff
$i_{2r-1}<i_{2s-1}<i_{2r}<i_{2s}$ or
$i_{2s-1}<i_{2r-1}<i_{2s}<i_{2r}$.
\end{itemize}

A {\it matching} is a subset of edges such that no two edges share a
common vertex. A vertex is said to be {\it saturated} if there is a
matching edge incident to it. A {\it perfect matching} is a matching
which saturates all vertices. There is a one-to-one correspondence
between the monomials in the Pfaffian and the perfect matchings in
$G$.

If $M$ is an $n\times n$ matrix and $A=\{i_i,\ldots ,i_r\} \subseteq
\{1,\ldots ,n\}$, then $M[A]$ denotes the matrix obtained from $M$
by deleting the rows and columns of indices in $A$. The {\it
Pfaffian Sum} of $M$ is a polynomial over indeterminates
$\lambda_1,\lambda_2,\ldots,\lambda_n$ defined by
\begin{displaymath}
\textrm{PfS}(M)=\sum_A(\prod_{i\in A} \lambda_i)\textrm {Pf} (M[A])
\end{displaymath}
where the summation is over the $2^n$ subsets of $\{1,\ldots ,n\}$.
There is a one-one correspondence between the terms of the Pfaffian
sum and the matchings in $G$. We consider only instances such that
each $\lambda_i$ is fixed to be 0 or 1. In this case, Pfaffian Sum
is a summation over all matchings that match all nodes with
$\lambda_i=0$. It is well known that both the Pfaffian and the
Pfaffian Sum are computable in polynomial time.

\subsection{Matchgate}

A {\it matchgate} $\Gamma$, is a quadruple $(G,X,Y,T)$, where
$G=(V,E,W)$ is a graph, $X\subseteq V$ is a set of input nodes,
$Y\subseteq V$ is a set of output nodes, and $T \subseteq V$ is  a
set of {\it omittable nodes} such that $X$, $Y$ and $T$ are pairwise
disjoint. Usually the numbers of nodes in $V$ are consecutive from 1
to $n=|V|$ and $X$, $Y$ have minimal and maximal numbers
respectively. Whenever we refer to the Pfaffian Sum of a matchgate
fragment, we assume that $\lambda_i=1$, if $i\in T$, and $0$
otherwise. Each  node in $X\cup Y$ is assumed to have exactly one
incident {\it external edge}. For a node in $X$, the other end of
the external edge is assumed to have index less than the index for
any node in $V$, and for a node in $Y$, the other end node has index
greater than that for every node in $V$. If $k=|X|=|Y|$, then
$\Gamma$ is called {\it $k$-bit matchgate}. A matchgate is called a
{\it level $k$ matchgate}, if it is an $n$-bit matchgate for some
$n\leq k$. If a matchgate only contains input nodes, output nodes
and one ommitable node, then it is called a {\it standard
matchgate}.

We define, for every $Z \subseteq X \cup Y$, the {\it character
$\chi(\Gamma,Z)$ of $\Gamma$ with respect to $Z$}  to be the value
$\mu (\Gamma,Z)\textrm{PfS}(G-Z)$, where $G-Z$ is the graph obtained
from $G$ by deleting the vertices in $Z$ together with their
incident edges, and the {\it modifier} $\mu (\Gamma,Z)\in \{-1,1\}$
counts the parity of the number of overlaps between matched edges in
$G-Z$ and matched external edges. We assume that all the nodes in
$Z$ are matched externally. By definition of the modifier, it is
easy to verify that $\mu (\Gamma,Z)=\mu (\Gamma,Z\cap X)\mu
(\Gamma,Z\cap Y)$, and that if $X=\{1,2,\ldots ,k\}$ and $Z\cap
X=\{i_1,i_2,\ldots ,i_l\}$, then $\mu (\Gamma,Z\cap X)=
(-1)^{\sum_{j=1}^l {(i_j-j)}}$.

The {\it character matrix} $\chi (\Gamma)$ is defined to be the
$2^{|X|} \times 2^{|Y|}$ matrix such that entry $(i_1 i_2 \ldots
i_k, i_{n} i_{n-1} \ldots i_{n-k+1})$ is $\chi(\Gamma,X' \cup Y')$,
where $X'=\{ j \in X | i_j=1 \}$,  $Y'=\{ j \in Y | i_j=1 \}$ and
$i_1 i_2 \ldots i_k$, $i_{n} i_{n-1} \ldots i_{n-k+1}$ are binary
expression of numbers between 0 and $2^k-1$. We also use $(X', Y')$
to denote this entry. We call an entry  $(X',Y')$  {\it edge entry},
if $0<|(X-X')\cup (Y-Y')|\leq2$. Throughout the paper, we identify a
matchgate and its character matrix. An easy but useful fact is that
for every $k$, the $2^k \times 2^k$ unit matrix is a character
matrix.

\subsection{Properties of character matrix}

We introduce several properties of character matrices, which will be
used in the proof of our results.

\begin{thm}[\cite{val02}]\label{prop:closeundermutiple}
If $A$ and $B$ are character matrices of size $2^k \times 2^k$, then
$AB$ is a character matrix.
\end{thm}

\begin{thm}[\cite{val02}] \label{prop:oneommitable}
Given any matchgate $\Gamma$ there exists another matchgate
$\Gamma'$ that has the same character as $\Gamma$ and has an even
number of nodes, exactly one of which is omittable.
\end{thm}

\begin{thm}[\cite{cc06c}] \label{prop:kehua}
Let $A$ be a $2^k \times 2^l$ matrix. Then $A$  is the character
matrix of a $k$-input, $l$-output matchgate, if and only if $A$
satisfies all the useful Grassmann-Pl\"ucker identities.
\end{thm}

This is a very useful  characterization of the character matrices
generalizing the characterization for a major part of all 2-input
2-output matchgates in \cite{val02}. The proof of this theorem
implies the following:

\begin{cor}[\cite{cc06c}] \label{cor:standardmg}
Let $A$ be a  $2^k \times 2^l$ matrix whose right-bottom most entry
is $1$ satisfying all the useful Grassmann-Pl\"ucker identities.
Then $A$ is uniquely determined  by its edge entries and $A$ is the
character matrix of a standard matchgate $\Gamma$ containing $k+l+1$
nodes ($k$ input nodes, $l$ output nodes and 1 omittable node).
\end{cor}

Recently, Cai and Choudhary also showed that:

\begin{thm}[\cite{cc06c}] \label{prop:twobitinvertable}
Let $A$ be a $4\times 4 $ character matrix. If $A$ is invertible,
then $A^{-1}$ is a character matrix. Consequently, the nonsingular
$4 \times 4$ character matrices form a group.
\end{thm}

\subsection{Matchcircuit}
Given a matchgate $\Gamma=(G,X,Y,T)$, we say that it is {\it even},
if $\textrm {PfS}(G-Z)$  is zero whenever $Z=X \cup Y$ has odd size,
and {\it odd} if $\textrm {PfS}(G-Z)$  is zero whenever $|Z|$ is
even.

\begin{thm}[\cite{val02},\cite{cc06c}]\label{thm:valiantandcai}
Consider a matchcircuit $\Gamma$ composed of gates as in
\cite{val02}. Suppose that every gate is:
\begin{enumerate}
\item a gate with diagonal character matrix,
\item an even gate applied to consecutive bits $x_i, x_{i+1},\ldots , x_{i+j}$ for some $j$,
\item an odd gate applied to consecutive bits $x_i, x_{i+1},\ldots , x_{i+j}$ for some $j$, or
\item an arbitrary gate on bits $x_1,\ldots, x_j$ for some $j$.
\end{enumerate}
Suppose also that every parallel edge above any odd matchgate, if
any, has weight $-1$ and all other parallel edges have weight 1.
Then the character matrix of $\Gamma$ is the product of the
character matrices of the constituent matchgates, each extended to
as many inputs as those of $\Gamma$.
\end{thm}

From now on, whenever we say a matchcircuit, we mean that it
satisfying the requirements in the above theorem. An example circuit
is shown in Fig. \ref{Figs:p4}, where the edges in a matchgate are
not drown, and each node has index smaller than that of all nodes
located to the right of the node.
  We call a matchcircuit
is of {\it level $k$}, if it is composed of matchgates no more than
$k$ bits.

 \begin{figure}[h]
 \begin{center}
 \includegraphics[width=0.45\textwidth]{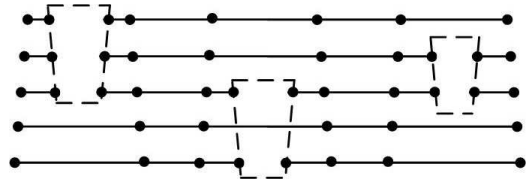}
 \caption{\label{Figs:p4}An example of matchcircuit.}
 \end{center}
 \end{figure}

 The {\it character matrix of a matchcircuit} is defined by the same way as that of  a matchgate except
that there is no modifier $\mu$.

\section{The results and overview of the proofs} \label{sec:tranformation T1-T4}

 \begin{thm} \label{thm:ourselfprop}
For every $k$, the nonsingular $2^k\times 2^k$ character matrices
form a group under the matrix multiplication.
 \end{thm}

We will prove theorem \ref{thm:ourselfprop} by induction on the size
of matchgates. The proof proceeds as follows. Based on corollary
\ref{cor:standardmg}, we observe that all $2^k\times 2^k$ character
matrices can be transformed to  a special form $2^{k}\times 2^{k}$
character matrices. This suggests the following:

\begin{definition}
We say that a $k$-bit matchgate is a {\it reducible matchgate}, if
the bottom pair of  nodes $k$ and  $n-k+1$ are connected by a weight
$1$ edge, and there is no other edge incident to any of the nodes
$k$ and $n-k+1$.

 The
character matrix of a reducible matchgate is called a {\it reducible
character matrix}.
\end{definition}

By corollary \ref{cor:standardmg}, a character matrix $B$ is a
reducible character matrix if it satisfies  the following:
\begin{enumerate}
\item [(1)] $B_{2^k-1,2^k-1}=B_{2^k-2,2^k-2}=1$.
\item [(2)] All the edge entries in the last two columns and the last two
rows are 0 except for $B_{2^k-2,2^k-2}$.
\end{enumerate}

Firstly we prove that if the $k$-bit nonsingular  character matrices
are closed under matrix inverse operation, then so are the
$(k+1)$-bit nonsingular reducible character matrices .

Secondly, we introduce some   elementary nonsingular matchgates so
that every nonsingular $2^{k}\times 2^{k}$ character matrix can be
transformed to a reducible character matrix by multiplying with the
character matrices of the elementary matchgates.

This transformation is realized by four phases as follows.  Starting
from $A=A^{(0)}$, we need the following:

\noindent {\bf Phase T1} ($A^{(0)} \Rightarrow A^{(1)}$). Turn the
right-bottom most entry to 1.

\noindent {\bf Phase T2} ($A^{(1)} \Rightarrow A^{(2)}$). Turn the
edge entries in the last row and column to 0's, while keeping the
right-bottom most entry 1.

\noindent {\bf Phase T3} ($A^{(2)} \Rightarrow A^{(3)}$). Turn the
entry $A^{(2)}_{2^k-2,2^k-2}$  to 1, while keeping the right-bottom
most entry 1 and the edge entries in the last row and column 0's.

\noindent {\bf Phase T4} ($A^{(3)} \Rightarrow A^{(4)}$). Turn the
edge entries in the row $2^k-2$ and column $2^k-2$ to 0's, while
 keeping the last two diagonal entries 1's and the edge entries
in the last row and column 0's.

Each phase consists of  several  {\it actions} (or for simplicity,
steps). In each step, either the positions of entries are changed,
or the values of some entries are changed.

An   action is defined to be the multiplication of a character
matrix with an elementary character matrix. The role of an action is
to change some specific entries to be some fixed value $0$ or $1$.
However, such an action will certainly injure other entries which
are undesired.

The crucial observation is that  an appreciate sequence of actions
will gradually satisfy all the entries requirements. During the
course of the transformation, once an entry requirement is satisfied
by some action, it will never be injured again by the future
actions. That is to say, an action may injure only the entries which
 have not been satisfied. This ensures that all the entries requirements
  will be eventually satisfied.

This describes the idea of the proof of theorem
\ref{thm:ourselfprop}. The proof  will also build an essential
ingredient for our second result, the theorem below.

\begin{thm} \label{thm:universalthm}
For every $k>2$, if $\Gamma$ is a matchcircuit composed of level $k$
matchgates, then:
\begin{enumerate}
\item [(1)] $\Gamma$ can be simulated by a level $2$ matchcircuit
$\Delta$.
\item [(2)]  A $k$-bit matchgate
can be simulated  by  O$(k^4)$ many single and two-bit matchgates.
And every matchcircuit $\Gamma$ can be simulated by a level $ 2$
matchcircuit in polynomial time.
\end{enumerate}
\end{thm}

Our proof of theorem \ref{thm:universalthm} is a composition of the
proof of theorem \ref{thm:ourselfprop} and some more elementary
matchgates. On the other hand, one could firstly prove theorem
\ref{thm:universalthm}, then prove \ref{thm:ourselfprop} by
combining theorem \ref{thm:universalthm} and theorem
\ref{prop:twobitinvertable}. However there are subtle difference
between character matrices of matchgate and matchcircuit. Therefore,
this approach needs additional technique.

\section{Group property of the $k$-bit character matrices}
\label{sec4}

In this section, we prove theorem \ref{thm:ourselfprop}.  To proceed
an inductive argument, we exploit the structure of the reducible
character matrices which pave the way to the reductions.

\subsection{Reducible matchgates}

\begin{lem} \label{lem:reductionlemma}
Let $\Delta_1$ be a $(k+1)$-bit reducible  matchgate, that is, the
bottom edge $(k+1, k+3)$ having weight $1$ and there is no any other
edge incident to any of the nodes $k+1$ and $k+3$. Let $\Gamma_1$ be
the $k$-bit matchgate obtained from $\Delta_1$ by deleting the
 edge $(k+1,k+3)$.
Then:

(i) If $\Delta_1$ is invertible, so is $\Gamma_1$.

(ii) If
 $\chi(\Gamma_1)^{-1}$ is  a
character matrix, so is  $\chi(\Delta_1)^{-1}$.
\end{lem}

\begin{proof}(Sketch)
For (i). This holds because $\chi(\Delta_1)$ is a block diagonal
matrix after rearranging the order of rows and columns, and
$\chi(\Gamma_1)$ is equal to one block.

For (ii). We prove this by constructing the inverted matchgate
$\Delta_2(F_2, W_2, Z_2, T_2)$ of $\Delta_1(F_1, W_1, Z_1, T_1)$
from the inverted gate $\Gamma_2(G_2, X_2, Y_2, T_2)$ of
$\Gamma_1(G_1, X_1, Y_1, T_1)$.

It suffices to prove that the composition of $\Delta_1$ and
$\Delta_2$ has the  unit matrix as its character matrix. See Fig.
\ref{Figs:p2} for the intuition of the proof, while detailed
verification will be given in the full version of the paper.

 \begin{figure}[h]
 \begin{center}
 \includegraphics[width=0.48\textwidth]{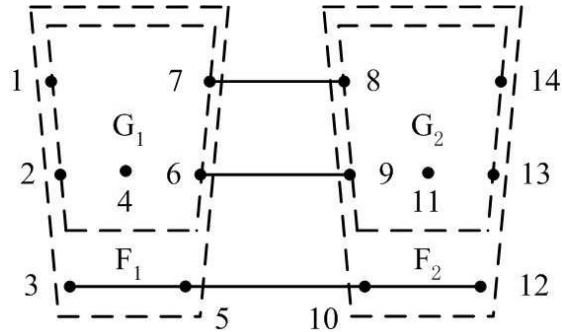}
\caption{Example of $k=2$. $X_1=\{1,2\}$, $Y_1=\{6,7\}$, $T_1=\{4\}$
$X_2=\{8,9\}$, $Y_2=\{13,14\}$, $T_2=\{11\}$, $W_1=\{1,2,3\}$,
$Z_1=\{5,6,7\}$, $W_2=\{7,8,9\}$, $Z_2=\{12,13,14\}$.
\label{Figs:p2}}
 \end{center}
 \end{figure}
\end{proof}

\subsection{The transformation lemma} \label{sec:transformation}
In this part, we construct the matchgates to realize the phases T1
-- T4 prescribed in section \ref{sec:tranformation T1-T4}, and show
that every $k$-bit nonsingular character matrix can be transformed
to a $k$-bit reducible character matrix by using the transformation.

The key point to the proof of the theorem is the following:

\begin{lem}\label{lem:changetoreducible}
Let $A$ be a $2^k \times 2^k$ nonsingular character matrix. Then
there exist nonsingular character matrices $L_s,\ldots, L_2,L_1,
R_1, R_2, \ldots, R_t$ for some $s$ and $t$ such that $L_s\cdots
L_2L_1A R_1$ $R_2 \cdots R_t$ is a reducible character matrix.
\end{lem}

\begin{proof}
Given a nonsingular character matrix $A$, we denote $A$ by
$A^{(0)}$. We construct the matchgates to realize the four phases T1
-- T4. We use $A^{(i)}$ to denote the character matrix obtained from
$A^{(i-1)}$ by using phase Ti, where $i=1,2,3,4$. We start with
$A^{(0)}$, and define the  transformation to be  a series of
actions, defined in section \ref{sec:tranformation T1-T4}. In the
discussion below, we will use $A$ to denote the character matrix
obtained so far in the construction from $A^{(0)}$ (or shortly, the
current matrix).

The four phases proceed as follows.

\noindent {\bf Phase T1:} Suppose that $\Gamma_l$ is the $k$-bit
matchgate such that the $l$-th pair of input-output nodes are
connected by a path of length 2 on which each edge has weight 1,
 and each of the other pairs is connected by an edge of weight $1$, and
$k+1$ is the only unomittable node other than the input and output
nodes. (See Fig. \ref{Figs:p3} (a)). Let $C_l$ denote the character
matrix of $\Gamma_l$.

 \begin{figure}[h]
 \begin{center}
 \includegraphics[width=0.45\textwidth]{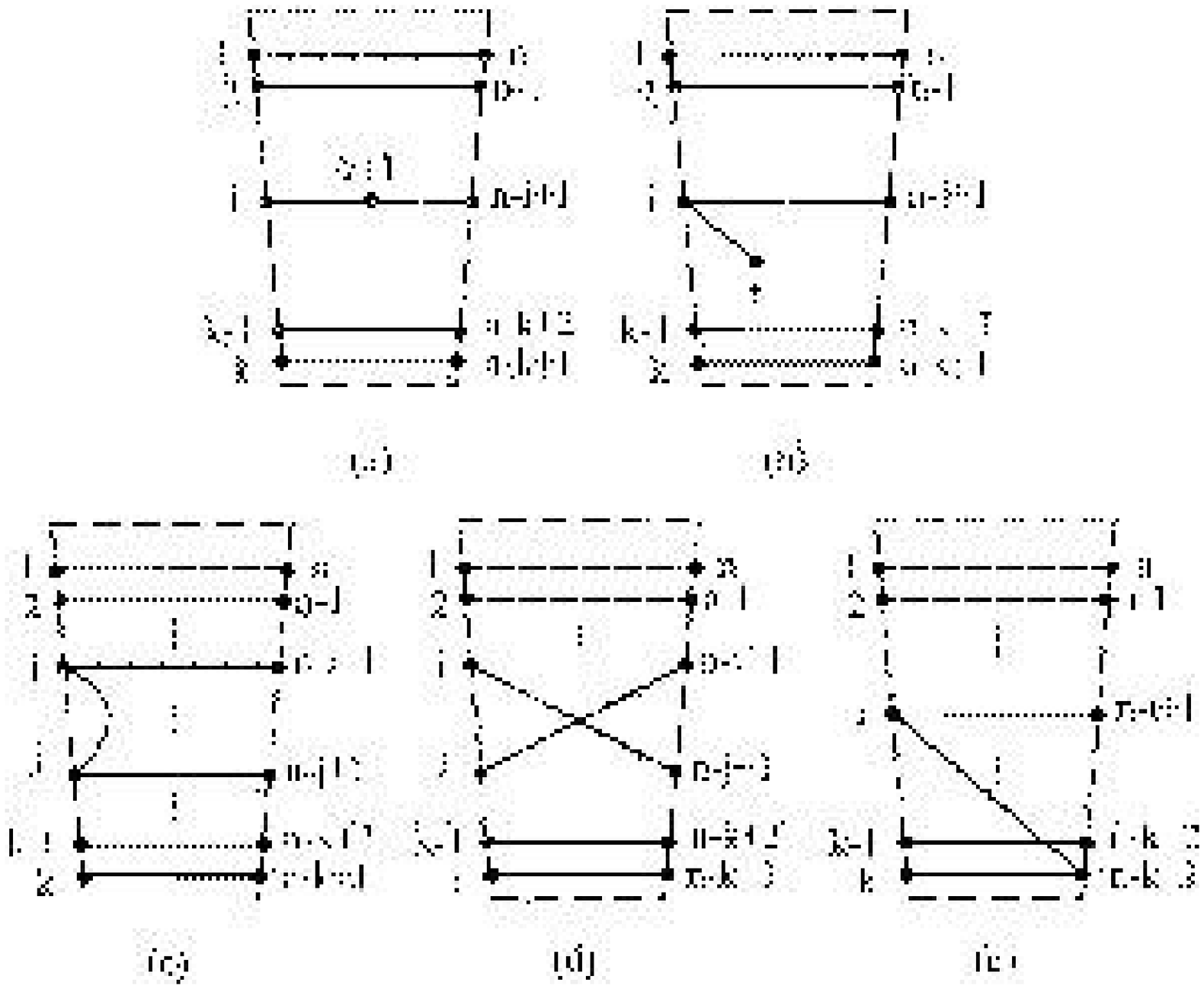}
\caption{ \label{Figs:p3}}
 \end{center}
 \end{figure}

Suppose without loss of the generality that $A_{I=i_1 i_2 \ldots
i_k,J=j_1 j_2 \ldots j_k} \neq 0$. Define

 $$L_1=\prod_{1 \leq l \leq
k, i_l=0} C_l,\ \ \ R_1=\prod_{1 \leq l \leq k, j_l=0} C_l.$$

It is easy to see that the right-bottom most entry, $a$ say, of $L_1
A R_1$ is either $A_{I,J}$ or $-A_{I,J}$ by computing the
$\textrm{PfS}(G-X \cup Y)$ of the composed matchgate corresponding
to $L_1 A R_1$. Let $L_2=\frac {1}{a} E$, and $A^{(1)}=L_2 L_1 A
R_1$. Clearly $L_2$ is a character matrix, so is $A^{(1)}$, by
theorem \ref{prop:closeundermutiple}.

\noindent {\bf Phase T2}: Phase T2 will change the edge entries in
the last column and the last row to 0's. We first describe the
actions for the column as follows.

{\bf Phase T2 for the last column}:

We turn the edge entries in the last column to $0$'s one by one from
bottom to top. To turn an edge entry $(X',2^k-1)$ to zero, we need a
row transformation applied to the current matrix $A$, which adds the
multiplication of $-b$ and the last row to row $X'$, where $b$ is
the value of entry $(X',2^k-1)$ of the current matrix.

Therefore phase T2 for the last column consists of the following
actions. In decreasing order of $X'$, for every edge entry
$(X',2^k-1)$, we have:

{\bf Action $(X',2^k-1)$}: Multiplying an elementary character
matrix, $L$ say,  to the current character matrix $A$ from the left
side, where $L$ is a character matrix satisfying that the diagonal
entries are all $1$'s, and that ${L}_{I=X',2^k-1}=-b$.

This  makes some row transformations to the current matrix according
to the nonzero entries other than the diagonal entries. The row
transformation corresponding to ${L}_{I=X',2^k-1}=-b$ is exactly the
one that realizes the goal of this action.


Now we formally construct the matchgate to realize the character
matrix $L$ as required in the action $(X',2^k-1)$ above. The
construction is divided into two cases depending on the size of $X'$
as follows.

{\bf Case 1}.  $X'=X-\{i\}$ for some $i$.

We use the matchgate with the following properties: (1) each
input-output pair of the gate is connected by an edge of weight $1$,
and (2) it contains one more edge $(i,t)$  to realize $L$, where $t$
is the unique omittable node, and the weight of $(i,t)$ is either
$b$ or $-b$ ensuring ${L}_{I,2^k-1}=-b$. For intuition of the
matchgate, a reader is referred to Fig. \ref{Figs:p3} (b).

Let $(I',J')$ be an arbitrary nonzero entry of $L$ other than the
diagonal entries. By the construction of the gate, we have that the
$i$-th bit of $I'$ and $J'$ are 0 and 1 respectively, and that $I'$,
and $J'$ are identical on the $j$-th  bit for every $j\not=i$. Hence
$I'<J'$ and $I' \leq I$ (recall that $I=X'$). The action at entry
$(X',2^k-1)$ in this case actually makes the following row
transformation:  For each such pair $(I',J')$, row $I'$ is added by
the multiplication of $L_{I',J'}$ and row $J'$. Since $I' \leq I$,
all the edge entries $(I_1,2^k-1)$ with $I_1>I$ have never been
injured by the action in this case.

{\bf Case 2}  $X'=X-\{i,j\}$ for some $i,j$.

The character matrix $L$ in this case is constructed by a similar
way to that in case 1 above, using the matchgate in Fig.
\ref{Figs:p3} (c).

The cost of the action in this case is similarly analyzed to that
for case 1.

Recall that after phase T1, the right-bottom most entry is  1. The
actions in both case 1 and case 2 of  phase T2  above will never
injure the last row of the matrix, so that the satisfaction of T1 is
still preserved by the current state of the construction.

{\bf Phase T2 for the last row}: The construction, and analysis for
the actions is the same as that for the column case with the roles
of rows and columns exchanged.

Therefore, the goal of T2 prescribed in section
\ref{sec:tranformation T1-T4} has been realized.

\noindent {\bf Phase T3}: The goal of this phase is similar to that
of phase T1, but different actions are needed. T3 consists of 2
actions. The first action moves a nonzero edge entry to position
$(2^k-1,2^k-1)$, and the second one changes edge entry
$(2^k-1,2^k-1)$ to 1. The actions proceed as follows.

{\bf Action 1}: First, we choose a nonzero edge entry. Since
$A^{(2)}$ is nonsingular, there must be a nonzero edge entry
$A^{(2)}_{X'=X-\{i\},Y'=Y-\{j\}}$ for some $i$ and $j$. (Otherwise,
all edge entries are zero's so that $A^{(2)}$ is a zero matrix,
contradicting the non-singularity of $A^{(0)}$.)

We use a gate of type $\Gamma_d$, defined as follows: (i) connect
each input-output pair other than the $i$-th or the $j$-th pair by
an edge, (ii) the $i$-th input  is connected to the $j$-th output,
and (iii) the $j$-th input is connected to the $i$-th output. All
edges are of weight 1. (See Fig. \ref{Figs:p3} (d))
 Let $C_{i,j}$ denote the character matrix of the matchgate
 described above.

This action just turns $A^{(2)}$ to $C_{i,k} A^{(2)} C_{j,k}$ by
connecting the  matchgate of $C_{i,k}$ with the gate of $A^{(2)}$,
and the gate of $C_{i,k}$ in the order of left to right.

Firstly, we verify that action 1 realizes its goal. Generally,
multiplying $C_{a,b}$ from left (resp. right) side is equivalent to
exchanging pairs of rows (resp. columns) $i_1i_2 \ldots i_a \ldots
i_b \ldots i_k$ and $i_1i_2 \ldots i_b \ldots i_a \ldots i_k$,
modular a factor of $1$ or $-1$. Hence, the edge entry
$(2^k-2,2^k-2)$ of $C_{i,k} A^{(2)} C_{j,k}$ is either
$A^{(2)}_{X',Y'}$ or $-A^{(2)}_{X',Y'}$.

Secondly, we analyze the cost of  the action. Notice that the row
exchanges are determined by a bit exchange on the labels of rows, so
that the number of zeros in (the string of) the row label is kept
unchanged. By definition, an edge entry can be exchanged only with
another edge entry. Therefore all edge entries in the last row and
column are kept zeros. In addition, it is easy to see that the
left-bottom most entry is kept 1.

{\bf Action 2}: We construct a matchgate with all of the
input-output pairs connected by an edge of weight 1, except that the
last pair is connected by an edge of weight $w= \frac
{1}{A_{2^k-2,2^k-2}}$.

All entries of the character matrix of this matchgate are zeros,
except for the diagonal entries. A diagonal entry $(I,I)$ is $w$, if
the last bit of $I$ is $0$, and $1$, otherwise.

We multiply this character matrix with the current matrix, then a
straightforward calculation shows that  entry $(2^k-1,2^k-1)$ is
turned to 1, while all the satisfied entries achieved previously are
still preserved.

The goal of T3 is realized.

\noindent{\bf Phase T4}: This phase is similar to phase T2, except
that we need consider the consequence on the last column and row. We
start from changing  the edge entries in column $2^k-2$.

{\bf Phase T4 for column $2^k-2$}:
 Suppose we are going to change
edge entry $(X-\{i\},Y-\{n-k+1\})$ to zero by the order from bottom
to top. Denote the action realizing this goal by {\it action at}
$(X-\{i\},Y-\{n-k+1\})$.

We construct the elementary matchgate used in the action at
$(X-\{i\},Y-\{n-k+1\})$. Each pair of input-output nodes of this
matchgate is connected by an edge of weight 1, furthermore, the
$i$-th input node is connected to the last output node by an edge of
weight $w$, where $w$ is either  $A_{X-\{i\},Y-\{n-k+1\}}$ or
$-A_{X-\{i\},Y-\{n-k+1\}}$ such that entry $(X-\{i\},Y-\{n-k+1\})$
of the character matrix of the matchgate is
$-A_{X-\{i\},Y-\{n-k+1\}}$. (See Fig. \ref{Figs:p3} (e).)

We examine the nonzero entries in the character matrix $L$ of the
constructed matchgate. We first   note that all diagonal entries are
1's. Let $(I',J')$ denote an arbitrary nonzero entry other than the
diagonal entries of the matrix $L$. By construction of the
matchgate, $I'$ and $J'$ differ at only the $i$-th and the $k$-th
bits, and $I'|_i=J'|_k=0$, $I'|_k=J'|_i=1$,  $I'<J'$, $I'<X-\{i\}$
and $I'$, $J'$ contain the same number of $0$'s, which is at least
$1$, where $I'|_i$ denotes the $i$-th bit of $I'$. The action at
$(X-\{i\},Y-\{n-k+1\})$ multiplies $L$  with $A$ from the left side.
It makes some row transformations: for every such entry $(I',J')$
chosen as above, add  row $I'$  by the multiplication of row $J'$ by
$L_{I',J'}$. So the goal of this action is realized.

Now we analyze the cost of the action. We first prove that it does
not injure the edge entries in column $2^k-2$ which have already
been satisfied. The reason is similar to that in phase T2. Because
$I'\leq X-\{i\}$, the action only injures the rows with indices less
than $X-\{i\}$.

The cost of the action is different from that in phase T2 in that it
may affect the edge entries in the last column which have already
been satisfied in phases T1 and T2. Because $I'$ and $J'$ contain
the same number of $0$'s, which is at least $1$, all the row changes
made by the action always add a zero edge entry of the last column
to another zero edge entry in the same column. Hence, it does not
injure the satisfied entries in the last column. Additionally, it is
obvious that  the last two rows are preserved during the current
action, so the left-bottom most entry, the edge entries in the last
row and entry $(2^k-2,2^k-2)$ are all preserved.

%
%

{\bf Phase T4 for row $2^k-2$}: Similar actions to that in phase T4
for the column above can be applied to the row $2^k-2$ to change its
edge entries to 0's.

Therefore, T4 realizes its goal, at the same time, it preserves the
satisfied entries in phases T1 -- T3.

We have realized the phases T1 -- T4 prescribed in section
\ref{sec:tranformation T1-T4}, by corollary \ref{cor:standardmg},
$B$ is a reducible character matrix. The lemma follows.
\end{proof}

\subsection{Proof of theorem \ref{thm:ourselfprop}}

\begin{proof} We prove by induction on $k$ that for
every $k$, and every $2^k\times 2^k $ character matrix $A$, if $A$
is invertible, then $A^{-1}$ is a character matrix.

The case for $k=1$ is easy, the first proof was given by Valiant in
\cite{val02}.

Suppose by induction that the theorem holds for $k-1$. By lemma
\ref{lem:changetoreducible}, there exist nonsingular character
matrices $L_i$ and $R_j$ such that $B=L_s \cdots L_2L_1A R_1 R_2
\cdots R_t$ is the character matrix of a reducible matchgate
$\Delta$. Let $B'$ be the $2^{k-1} \times 2^{k-1}$ character matrix
of $\Gamma$ constructed from $\Delta$ by deleting the bottom edge.

Since $A$ is invertible, so is $B$, and so is $B'$ by lemma
\ref{lem:reductionlemma}. By the inductive hypothesis, $B'^{-1}$ is
a character matrix, so is $B^{-1}$ by lemma
\ref{lem:reductionlemma}.

By the choice of $L_i$ and $R_j$, for all $1\leq i\leq s$ and $1\leq
j\leq t$, we have that

 $$A^{-1}=R_1R_2 \cdots R_t B^{-1} L_s \cdots L_2 L_1.$$

 By theorem
\ref{prop:closeundermutiple}, $A^{-1}$ is also a character matrix.

This completes the proof of theorem \ref{thm:ourselfprop}.
\end{proof}

We notice that the inductive argument in the proof of theorem
\ref{thm:ourselfprop} also gives a different proof for the result in
the case of $k=2$. Our method is a constructive, and uniform one. It
may have some more applications.

\section{Level $2$ matchgates are universal} \label{sec5}

We introduce nine types of  matchgates as our elementary gates.  We
use $\Gamma_a, \ldots, \Gamma_i$, to denote the elementary level $2$
matchgates corresponding to that in the following Fig. \ref{Figs:p5}
(a), (b), (c), (d),  (e), (f), (g), (h) and (i) respectively.

 \begin{figure}[h]
 \begin{center}
 \includegraphics[width=0.45\textwidth]{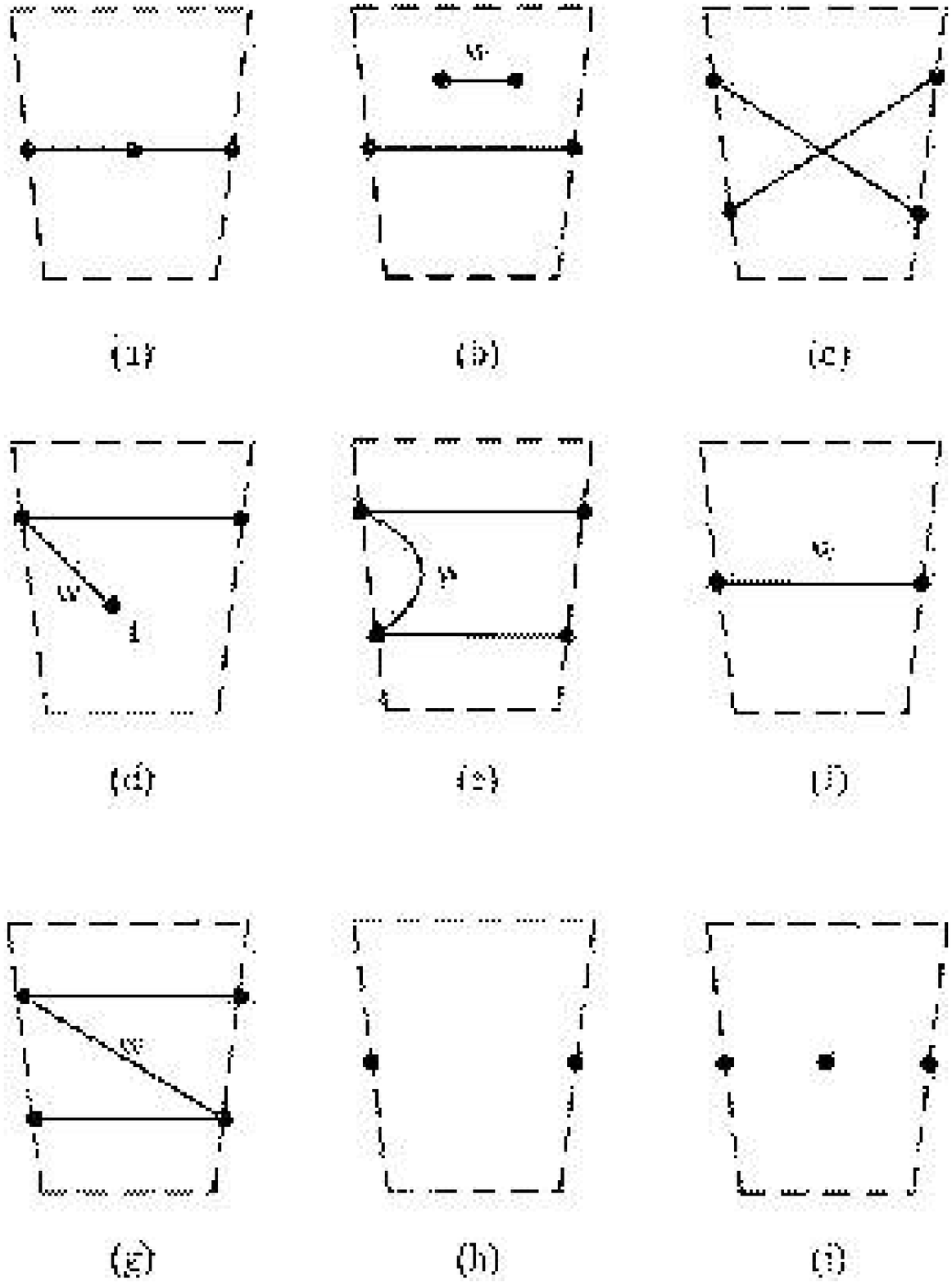}
\caption{ \label{Figs:p5}}
 \end{center}
 \end{figure}

We describe the elementary gates as follows. All edges in $\Gamma_a$
have weight 1. All edges connecting an input and an output node
except for the edge in $\Gamma_f$, and the diagonal edge in
$\Gamma_g$, are all of weight  1. The remaining edges take
 weights $w$.

$\Gamma_a$ makes a row (or column, when it is multiplied from right
side) exchange, which is a special transformation, of the character
matrix according to a bit flip on the label, and it is used to move
a nonzero entry to the right-bottom most entry by the same way as
that in the proof of theorem \ref{prop:kehua} in \cite{cc06c}.
$\Gamma_b$ is used to realize $cE$, and to turn a nonzero entry to
1. Both $\Gamma_a$ and $\Gamma_b$ are only used in the first phase,
i.e. T1, of the transformation. Intuitively, $\Gamma_c$ can exchange
two consecutive bits, and it allows us to apply some other
elementary gates to nonconsecutive bits. $\Gamma_d$ and $\Gamma_e$
are used in  phase T2 to eliminate the edge entries in the last
column and the last row. $\Gamma_c$ will be also used in phase T3 to
move a nonzero edge entry to position $(2^k-2,2^k-2)$, in which
case, $\Gamma_f$ will further turn this entry to 1. $\Gamma_g$ is
used in phase T4 to eliminate the edge entries in the column $2^k-2$
and row $2^k-2$. A nonzero singular character matrix will be
transformed to a matchcircuit composed of only $\Gamma_h$-type
gates. $\Gamma_i$ is used to realize zero matrix. To understand the
composition of $\Gamma_c$ with other elementary gates, we need the
following:

\begin{lem} \label{lem:effectofgamma_c}
Suppose $A$ is the character matrix of a $k$-bit matchcircuit
$\Delta$, and $P_1$, $P_2$  are two arbitrary permutations on $k$
elements. There exists matchcircuit $\Lambda$ constructed from
$\Delta$ by adding some gates $\Gamma_c$, such that the
corresponding character matrices $B$ satisfying
$B_{2^k-1,2^k-1}=A_{2^k-1,2^k-1}$ and $|B_{i_1 \cdots i_k, j_1
\cdots j_k}|=|A_{P_1(i_1 \cdots i_k), P_2(j_1 \cdots j_k)}|$.
\end{lem}

The following lemma gives the transformation for matchcircuits.

\begin{lem}\label{lem:circuittransform}
For any $k>2$, and any $k$-bit matchcircuit $\Delta$ consisting of a
single nonsingular $k$-bit matchgate $\Gamma$, there is a new
matchcircuit $\Lambda$ constructed by adding some invertible single
 and two-bit matchgates to $\Delta$, such that the character
matrix $B$ of $\Lambda$ is reducible. Furthermore, $B$ is the
character matrix of an even reducible matchgate.
\end{lem}

So far we have established the result for the first significant case
that a matchgate is applied to the first $k$ bits.

In the following lemma we consider two more cases:
\begin{itemize}
\item a gate  applied to consecutive
bits but not starting from the first bit,
\item a gate applied to nonconsecutive bits.
\end{itemize}

 For the first case, the gate must be an even or an odd gate, we
 observe
that only even and odd gates are used in the transformation for an
even or an odd gate. For the second case, we extend its matrix, and
replace it by a new  even gate which is  applied to consecutive bits
reducing it to the first case.

\begin{lem} \label{lem:univeral}
 For any $k>2$, and any $m$-bit matchcircuit $\Delta$ containing a  $k$-bit
matchgate $\Gamma$ with character matrix $A$, there is a level $k-1$
matchcircuit $\Lambda$ having the same character matrix as $\Delta$.
\end{lem}

The proof for lemma \ref{lem:effectofgamma_c}-\ref{lem:univeral}
will be given in the full version.

\subsection{Proof of theorem \ref{thm:universalthm}}

\begin{proof} For (1). Repeat the process in lemma \ref{lem:univeral}  until there
is no gate of bit greater than $2$.

For (2). The number of matchgates used in the phases of
transformation are O$(k)$, O$(k^3)$, O$(k)$ and O$(k^2)$,
respectively, so a $k$-bit matchgate can be simulated by  O$(k^4)$
many single and two-bit matchgates. This procedure is polynomial
time computable, because there are polynomially many actions, and
each action is polynomial time computable due to the fact that we
compute only the edge entries.
\end{proof}

%

\end{document}